\documentclass[letterpaper, 10 pt, conference]{ieeeconf}  
\IEEEoverridecommandlockouts                              
\usepackage{mathtools}
\usepackage{bm}
\usepackage{bbm}

\usepackage{graphicx}
\usepackage{amssymb}
\usepackage{color}

\newcommand{\R}{\mathbb{R}} 

\newcommand{\CA}{\mathcal{A}} %
\newcommand{\CC}{\mathcal{C}} %
\newcommand{\CZ}{\mathcal{Z}} %

\newcommand{\dist}{{\rm dist}}

\newcommand{\col}{{\rm col}}

\usepackage{amsthm}

\theoremstyle{definition}
\newtheorem{definition}{Definition}

\newtheorem{lemma}{Lemma}

\theoremstyle{plain}
\newtheorem{theorem}{Theorem}

\newcommand*{\QE}{\hfill\ensuremath{\square}}%
\newenvironment{pfof}[1]{\vspace{1ex}\noindent{\textit{Proof of
		#1:}}\hspace{0.5em}} {\hfill\QE\vspace{1ex}}

\usepackage{xcolor}
\definecolor{myblue}{HTML}{e0f3f8}
\definecolor{mygreen}{HTML}{abdda4}
\definecolor{mylightpurple}{HTML}{fde0ef}
\definecolor{darkgreen}{HTML}{4d9221}

\usepackage{cite}
\usepackage{amsmath}
\usepackage{tikz}
\usepackage{mathtools}

\usepackage[colorlinks]{hyperref}
\hypersetup{citecolor=[RGB]{5,112,176},linkcolor=[RGB]{5,112,176},urlcolor=[RGB]{5,112,176}}

\title{\LARGE \bf Analytical Characterization of Epileptic Dynamics in a Bistable System}

\author{Yuzhen Qin, Ahmed El-Gazzar, Danielle S. Bassett,  Fabio Pasqualetti, and Marcel van Gerven
  \thanks{This work was supported in part by the project Dutch Brain Interface Initiative (DBI$^2$) with project number 024.005.022 of the research programme Gravitation which is (partly) financed by the Dutch Research Council (NWO). It is also supported in part by NSF CMMI-2308639.
    Y. Qin, A. El-Gazzar, and M. van Gerven are with the Donders Institute for Brain, Cognition and Behaviour, Radboud University, Nijmegen, the Netherlands.  F. Pasqualetti is with the Department of Mechanical
	Engineering, University of California at Riverside. D. S. Bassett is with the Department of Bioengineering, the Department of Electrical \&
	Systems Engineering, the Department of Physics \& Astronomy, the
	Department of Psychiatry, and the Department of Neurology, University
	of Pennsylvania, and The Santa Fe Institute. (Contact: yuzhen.qin@donders.ru.nl).  
}
}

\begin{document}

\maketitle
\thispagestyle{empty}
\pagestyle{empty}

\begin{abstract}
    Epilepsy is one of the most common neurological disorders globally, affecting millions of individuals. Despite significant advancements, the precise mechanisms underlying this condition remain largely unknown, making accurately predicting and preventing epileptic seizures challenging. In this paper, we employ a bistable model, where a stable equilibrium and a stable limit cycle coexist, to describe epileptic dynamics. The equilibrium captures normal steady-state neural activity, while the stable limit cycle signifies seizure-like oscillations. 
    The noise-driven switch from the equilibrium to the limit cycle characterizes the onset of seizures. The differences in the regions of attraction of these two stable states distinguish epileptic brain dynamics from healthy ones. We analytically construct the regions of attraction for both states. Further, using the notion of input-to-state stability, we theoretically show how the regions of attraction influence the stability of the system subject to external perturbations. Generalizing the bistable system into coupled networks, we also find the role of network parameters in shaping the regions of attraction. Our findings shed light on the intricate interplay between brain networks and epileptic activity, offering mechanistic insights into potential avenues for more predictable treatments.
\end{abstract}


\section{INTRODUCTION}
Epilepsy affects over 65 million individuals worldwide \cite{2022_BC_WRG_et_al}. Patients with epilepsy usually experience recurring seizures characterized by the sudden emergence of excessively synchronous neural activity \cite{201_JP_DCM_JJGR_et_al}. Anti-seizure medications, relying on stabilizing electrical activity in the brain to reduce the likelihood of seizures, are the most widely used treatment for epilepsy. However, one-third of the patients remain resistant to such drugs \cite{2020_LW_PH_SSM_VA}. Neurostimulation and surgery are alternative treatments. Yet, only 50\% of patients after surgical resection have long-term freedom of seizures. When it comes to neurostimulation, the success rate is even lower. Nearly 95\% of individuals under vagus nerve stimulation therapies continue to have seizures \cite{2015_MSL_PE_RP}. 

One of the main factors for the limited success of current treatments is the lack of understanding of the precise mechanisms underlying epilepsy as a network disease \cite{2013_GEM_CDA_Nat_Review}. Dynamical system modeling is promising since it offers the potential to uncover the mechanisms underlying emerging behaviors associated with epileptic seizures and subsequently identify specific targets for intervention \cite{2012_SRA_SRG_TSS}. 

\textbf{Related work.} Various models have been proposed to capture epileptic dynamics \cite{2016_WF_BP_BF_JV}.  Many of them use two fundamental states to describe epileptic activity in the brain or brain regions: an \textit{interictal} one characterized by normal steady-state neural activity, and an \textit{ictal} one marked by excessively synchronous oscillations \cite{2003_DS_FL_et,2003_DSFGP_BW_et_al}. In the realm of dynamical systems, these states correspond to equilibrium and limit cycle, respectively. Some studies focus on characterizing how parameter alterations destabilize the equilibrium, leading systems to seizure-like oscillations. In \cite{2019_NE_CJ_ACC,2022_NE_PR_CJ:Auto,2023_MM_MT_CJ:LCSS}, researchers consider networks of linear-threshold units and construct conditions under which the limit cycle is the only stable state. Using the same model, the work \cite{2021_CF_AA_PF_CJ:LCSS} analytically investigates how bifurcations induced by external inputs engender pathological oscillations. Subsequent research endeavors propose network control strategies to contain the spread of such oscillations \cite{2022_AA_CF_PF_CJ:OJCS} or restore desired oscillations \cite{2022_own_CDC_a,2023_own_OJCS,2023_own_ACC}. Additionally, other investigations link epileptic dynamics with instability in discrete-time linear fractional-order systems \cite{2023_REA_RG_BP_PS_ACC}. 

The studies above primarily focus on exploiting the instability of equilibrium within systems to describe epileptic activity. Alternative approaches propose models that allow for bistability. This bistability is ideal for capturing the transition from a stable equilibrium to a stable limit cycle, which signifies the onset of epileptic seizures, wherein normal steady-state neural activity gives way to excessive oscillations \cite{2010_FF_STJ_BM,2012_BG_GM_et}. These transitions can be induced by endogenous or exogenous perturbations. A simple model with two stable states was introduced to capture similar transitions \cite{2010_KSN_VDN_epilepsy_bistable}. Subsequent research efforts have extended this model to networks of bistable units, exploring the influence of network structure on transition possibilities through numerical investigations \cite{2012_BO_FTHB_Math_epilepsy}.
Inspired by these studies, we utilize this bistable model as a mesoscale building block for describing firing-rate neural dynamics. 

\begin{figure*}[t]
		\centering
		\includegraphics[scale=0.4]{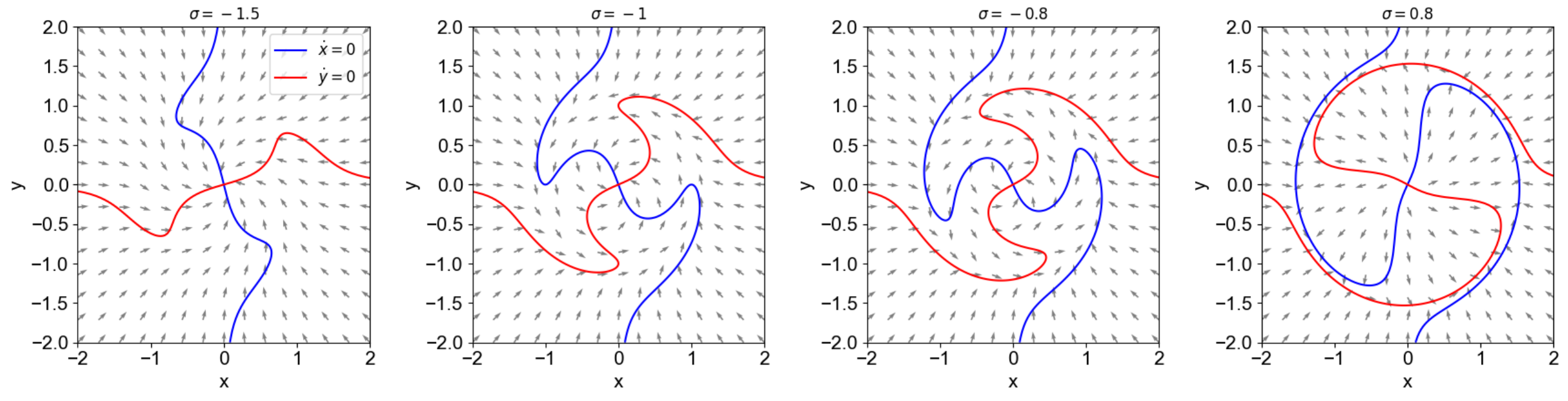}
		\caption{Nullclines and vector fields of the system \eqref{main} in different parameter regimes. In all cases, $a=1,b=1$, and $\omega = 0.4$.}
		\label{nullclines}
	\end{figure*}
 
\textbf{Contribution.}  Our contributions are threefold.  Firstly, employing Lyapunov methods and LaSalle's invariance principle, we rigorously identify the regions of attraction for both the stable equilibrium and limit cycle within the bistable system.
We put forth that disparities in these regions of attraction serve as distinguishing factors between epileptic and healthy brains. Systems with smaller regions of attraction towards the equilibrium are more prone to transitioning into pathological oscillations in the presence of noise. Secondly, by leveraging the concept of input-to-state stability, we construct conditions for external perturbations to ensure the system remains in proximity to the equilibrium. This also formally demonstrates that systems with smaller regions of attraction are more vulnerable to perturbations. Finally, we investigate a network comprising bistable units and analytically estimate the overall network's regions of attraction. This analysis sheds light on the interplay between network parameters and the network's susceptibility to perturbations.

In addition to our theoretical results, we discover that networks composed of bistable units can accurately replicate EEG recordings from epileptic brains through parameter training. Notably, this approach, utilizing bistable dynamics as a generative model, holds significant promise in understanding the network pathology of epilepsy. Together with our analytical findings, this method could offer valuable insights into potential objectives for controlling epileptic dynamics, thereby informing the design of targeted, predictive, and closed-loop treatments.

\textbf{Notation:} Given $x\in\R^n,y\in\R^m$, define $\col(x,y) \coloneq [x^\top,y^\top]^\top$. A function $\beta\colon [0,a)\times [0,\infty)\to [0,\infty)$ is called a class $\mathcal{KL}$-function if for any fixed $s$, $\beta(a,s)$ is increasing and, for any fixed $a$, $\beta(a,s)$ is decreasing and $\lim_{s\to\infty} \beta(a,s)=0$. Given a set $\CC\subset \R^n$, the distance from a point $x\in \R^n$ to $\CC$ is denoted as $\dist(x,\CC) := \inf_{y \in \CC}\|x-y\|$.

\section{Problem Formulation}
We first introduce the following model to describe the activity of individual neuronal populations:
\begin{subequations}\label{main}
    \begin{align}
    \dot x &= - \omega y + x \left(\sigma + 2 a b (x^2+y^2)-b(x^2+y^2)^2 \right) \,, \\
    \dot y &= \omega x + y \left(\sigma + 2 ab (x^2+y^2)-b(x^2+y^2)^2 \right) \,,
\end{align}
\end{subequations}
where $x,y\in \R$, $\omega, a, b>0$, and $\sigma\in \R$.

This model is an equivalent variation of the one found in \cite{2010_KSN_VDN_epilepsy_bistable}.  
Previous work has demonstrated that special versions of the system exhibit different behaviors as the parameters $\omega, \sigma, a, b$ vary \cite{2018_Srogatz_nonlinear_dyna, 2012_BO_FTHB_Math_epilepsy}. We start by presenting a comprehensive characterization of the possible behaviors exhibited by the system \eqref{main} across various parameter regimes.

First, let us rewrite the system into polar coordinates. Let $re^{i\theta} = x+ iy$, then we arrive at
\begin{subequations}\label{dyna_r}
\begin{align}
    &\dot r = r (\sigma +2 ab r^2-br^4)\,,\label{dyna_r2}\\
    & \dot \theta = \omega\,.
\end{align}
\end{subequations}
The system \eqref{dyna_r} has at most three possible equilibria
\begin{equation*}
    r=0\,, \hspace{5pt} r = \sqrt{a+\gamma_0}\,, \hspace{5pt} r = \sqrt{a - \gamma_0}\,,
\end{equation*}
with $\gamma_\mu:= \sqrt{a^2 + (1-\mu)\sigma/b}$ for any $\mu\in[0,1]$, which respectively correspond to the equilibrium $\col(x,y) = 0$, denoted as $e_0$,  and the limit cycles,
\begin{align*}
    &\CC_1 \coloneq\{x,y\in\R\colon x^2 +y^2 = {a + \gamma_0}\}\,,\\
    &\CC_2\coloneq\{x,y\in\R\colon x^2 +y^2 = {a -\gamma_0}\}
\end{align*}
of the original system \eqref{main}. Note that not all of these potential equilibria coexist and are stable simultaneously. We next characterize their stability at different parameter regimes.

\begin{lemma}\label{local_stability}
    For the system \eqref{main}, the following statements hold:
    \begin{enumerate}
        \item[(i)] if $\sigma < -a^2b$, the system has a unique stable equilibrium $\col(x,y)=0$.
        \item[(ii)] if $ \sigma=-a^2b$, the system has a stable equilibrium $\col(x,y)=0$ and a semi-stable limit cycle $\CC_1$.        
        \item[(iii)] if $-a^2b < \sigma < 0$, the system has a stable  equilibrium $\col(x,y)=0$, a stable limit cycle $\CC_1$, and an unstable limit cycle $\CC_2$.     
        \item[(iii)] if $\sigma \ge 0$, the system has a stable limit cycle $\CC_1$ and unstable equilibrium $\col(x,y)=0$.
    \end{enumerate}
\end{lemma}

The nullclines and vector fields of the system \eqref{main} in Fig.~\ref{nullclines} provide a more intuitive illustration of the situations in this lemma. Most interestingly, in Case (iii), the system exhibits a \textit{bistable} behavior, where 
a stable limit cycle and a stable equilibrium coexist. Around the equilibrium, the system behaves as a damped oscillator; on the limit cycle, the system is a self-sustained oscillator with amplitude $\sqrt{a+\gamma_0}$ and frequency ${2\pi}/{\omega}$. Note that $\sigma$ serves as the bifurcation parameter, whereas $a$ primarily governs the oscillation amplitude of the limit cycles and $\omega$ determines the oscillation frequency. We refer to the system \eqref{main} as a bistable oscillator in the remainder of this paper.

 In the bistable regime of $-a^2 b <\sigma<0$, the system is particularly suitable for describing epileptic dynamics. The stable equilibrium at the origin describes `background' neural activity at a normal state without epileptic seizures.  Conversely, the stable limit cycle characterizes the oscillating neural activity observed during pathological conditions in the presence of seizures (known as ictal states~\cite{2012_BO_FTHB_Math_epilepsy}).  The shift from equilibrium to a limit cycle captures the transition from a normal steady state to seizure activity.  External inputs or internal noise serves as the driving force behind the switch between these two stable states.


In this paper, we aim to investigate the chance of such transitions by rigorously studying the regions of attractions of the above two states and investigating how the system responds to external inputs. Additionally, we delve into networks comprising these bistable units, exploring the influence of network parameters on the regions of attraction.

\section{Individual Bistable Oscillators}
\subsection{Regions of attraction}
\begin{definition}
    Given a system $\dot z = f(z,t), z\in\R^n$ that has a stable equilibrium $e_0$ and a stable limit cycle $\CC$, let $\phi(t,z)$ be the solution of this system that starts at the initial state $z$. The basin of attraction of $e_0$ is defined by
    \begin{equation*}
        \CA_e \coloneq \{z \in \R^n\colon \lim_{t\to \infty} \phi(t,z) = e_0\}\,.
    \end{equation*}
    The basin of attraction of  $\CC$ is defined by 
    \begin{equation*}
        \CA_\ell\coloneq \{z \in \R^n\colon \lim_{t\to \infty} \dist(\phi(t,z),\CC) \to 0\}\,.
    \end{equation*}
\end{definition}
    
Next, we study the regions of attraction of the equilibrium $\col(x,y)=0$ and the limit cycle $\CC_1$ for the system \eqref{main}. 

\begin{theorem}[Regions of Attraction]\label{stability}
    For the system \eqref{main}, consider the bistable regime where $-a^2b < \sigma < 0$. Then, the following statements hold:\\    
    (i) The region of attraction of $\col(x,y)=0$ is
    \begin{equation*}
        \CA_e\coloneq\{x,y\in\R \colon x^2+y^2 < {a-\gamma_0}\}\,.
    \end{equation*}
    (ii) The region of attraction of the limit cycle $\CC_1$ is
    \begin{equation*}
         \CA_\ell \coloneq \{x,y\in\R \colon x^2+y^2 > {a-\gamma_0}\}\,.
    \end{equation*}
 \end{theorem}

    \begin{figure}[t]
    \centering
    \begin{tikzpicture}[scale=1.3]
    \draw[,line width=1] (0,0) circle[radius=1];
    \fill[fill=myblue] (0,0) circle[radius=0.5];
    \draw[dashed,line width=1] (0,0) circle[radius=0.5];
    
    \filldraw[black] (0,0) circle (0.04);
    \draw[->,dotted,line width= 1] (-1.5,0) -- (1.5,0) node[below] {$x$};
    \draw[->,dotted,line width= 1] (0,-1.5) -- (0,1.5) node[left] {$y$};
    
    \node at (-1.2, 1.3) {$(a)$};
    \node at (2.3, 1.3) {$(b)$};


    \draw[-stealth,blue] plot[smooth, tension=0.8] coordinates {(0.3,0.3) (0.1,0.2) (0,0)};

    \draw[-stealth,red] plot[smooth, tension=0.8] coordinates {(-0.10,1.4)  (-0.4,1.2) (-0.8,0.62)};
    \draw[-stealth,red] plot[smooth, tension=0.8] coordinates {(0.7,-0.4)  (0.8,0.3) (0.7,0.7)};

    \draw[,line width=1] (3.5,0) circle[radius=1];
    \fill[fill=myblue] (3.5,0) circle[radius=0.6];
    \draw[dashed,line width=1] (3.5,0) circle[radius=0.6];

    \fill[fill=mylightpurple] (3.5,0) circle[radius=0.3];
    \draw[dashed,line width=1,red] (3.5,0) circle[radius=0.3];

    \draw[->,dotted,line width= 1] (2,0) -- (5,0) node[below] {$x$};
    \draw[->,dotted,line width= 1] (3.5,-1.5) -- (3.5,1.5) node[left] {$y$};
    \filldraw[black] (0,0) circle (0.04);

    \draw[-stealth,line width= 0.6,blue] (3.5,0) -- (3.75,.4) node[right] {noise};
    \filldraw[blue] (3.75,.4) circle (0.02);
    \end{tikzpicture}
    \caption{Illustration of bistability when $-a^2b<\sigma<0$. (a) A stable equilibrium (the origin) and a stable limit cycle (solid circle) coexist. The region of attraction for the origin is depicted by the shaded open area. Starting from inside, the system converges to the origin; starting from outside, it converges to the limit cycle. The \textit{separatrix} is the dashed circle. (b) Different regions of attractions provide a possible interpretation of the discrepancy between healthy and epileptic brains.}
    \label{bistable_illustration}
    \end{figure}
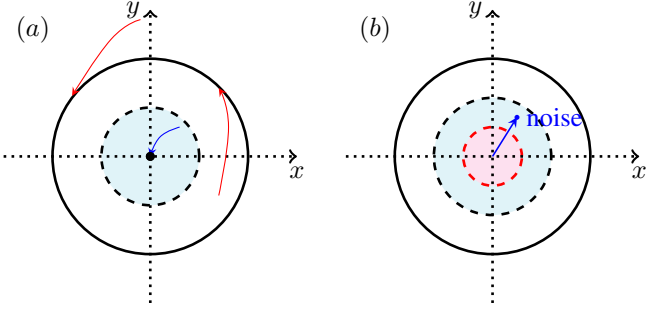

    Theorem~\ref{stability} provides a rigorous characterization of the regions of attraction for limit cycles and equilibria. An intuitive illustration is shown in Fig.~\ref{bistable_illustration}a. The unstable limit cycle $\CC_2$, depicted by the dashed circle in Fig.~\ref{bistable_illustration}a, is the separatrix of the two stable modes of behavior. The radius of the separatrix is given by $a-\gamma_0=a-\sqrt{a^2 + \sigma/b}$.  For fixed $a$ and $b$, as $\sigma$ approaches 0, the radius decreases; as $\sigma$ approaches $a^2b$, the radius approaches $a$.

    From a dynamical systems perspective, we posit that the different sizes or shapes of regions of attraction distinguish an epileptic brain from a healthy one. In a healthy brain, the region of attraction for the equilibrium is sufficiently large, thus making it unlikely for neural activity to transition into oscillations in response to noise. In contrast, the epileptic brain possesses a smaller region of attraction, rendering it more susceptible to noise. Even under the influence of the same level of noise, the epileptic brain is more prone to transitioning into oscillations (see Fig.~\ref{bistable_illustration}b). 

\subsection{Forced Bistable Oscillators}
Given the noisy nature of the brain, fluctuations often exert continuous influence on the neural activity. Motivated by this observation, we consider the forced model given by
\begin{subequations}\label{forced}
    \begin{align}
    &\dot x = - \omega y + x \left(\sigma + 2 ab(x^2+y^2)-b(x^2+y^2)^2 \right)+u_1(t)\,, \\
    & \dot y = \omega x + y \left(\sigma + 2 ab(x^2+y^2)- b (x^2+y^2)^2 \right)+u_2(t)\,,
\end{align}
\end{subequations}
where the inputs $u_1(t)$ and $u_2(t)$ are used to model perturbations to the system, such as noise and external stimuli. Denote $u \coloneq \col(u_1,u_2)$.  We assume that the system \eqref{main} operates in the bistable regime where $-a^2 b <\sigma<0$.

Ideally, we want the system to stay around the equilibrium, even when subjected to external perturbations, as the equilibrium corresponds to the normal state in the brain.
Next, we rigorously explore the criteria that initial conditions and inputs should meet to achieve this objective.

\begin{theorem}[Input-to-State Stability]\label{Theo_ISS}
     For any $\mu \in(0,1)$, let 
     \begin{equation}\label{bound_initial}
         \mathcal{B}_{\mu} \coloneq \{x,y\colon x^2 +y^2<a-\gamma_\mu \}\,,
     \end{equation} 
     where $\gamma_\mu:= \sqrt{a^2 + (1-\mu)\sigma/b}$. For any initial condition $\col(x(0),y(0))\in \mathcal{B}_{\mu}$ and any input $u$ satisfying
     \begin{align}\label{bound_inputs}
         \sup_{t\ge 0}\|u(t)\|< (1-\varepsilon)\mu |\sigma|\sqrt{a-\gamma_\mu}
     \end{align}
     for any $\varepsilon\in(0,1)$, the solution to the system \eqref{forced} satisfies
     \begin{align}\label{input_to_state}
         &\|\col(x(t),y(t))\|\nonumber\\
         &\le \beta \left( \left\|\col\big(x(0),y(0)\big) \right\|,t \right) + \frac{1}{(1-\varepsilon)\mu|\sigma|} \sup_{t\ge 0}\|u(t)\|\,,
     \end{align}
     where $\beta(\cdot)$ is a $\mathcal {KL}$-function.
\end{theorem}

\begin{figure}[t]
		\centering
		\includegraphics[scale=0.8]{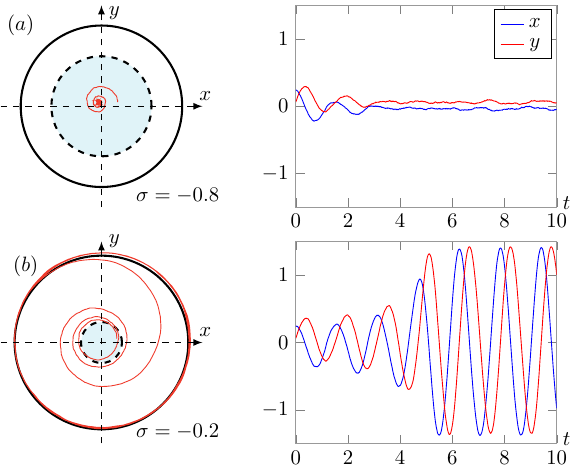}
		\caption{Evolution of the system \eqref{forced} at different regimes with the same initial condition and identical continuous perturbations ((a): $\sigma = -0.8$, (b): $\sigma = -0.2$). In both cases, $a=1,b=1, \omega = 4$. Perturbations are Gaussian with mean $0.2$ and STD $0.5$ truncated at $[-0.3,0.7]$. }
		\label{noisy_osci}
\end{figure} 

\begin{proof}
    Consider a continuously differentiable function $V(x,y)=\frac{1}{2}(x^2+y^2)$. Its time derivative satisfies
    \begin{align*}
        \dot V(x,y) &=(x^2+y^2)\left(\sigma + 2 ab (x^2+y^2)-b(x^2+y^2)^2 \right)\\
        &\hspace{4cm}+ xu_1(t)+yu_2(t)\\
        &\le (x^2+y^2)\left(\sigma + 2 ab(x^2+y^2)-b(x^2+y^2)^2 \right) \\
        &\hspace{4cm}+\|u\|\sqrt{x^2+y^2}\,,
    \end{align*}
    where the Cauchy–Schwarz inequality has been used. 
    
    Next, we bound the term $\sigma + 2 ab(x^2+y^2)-b(x^2+y^2)^2$. To do that, we define a function $f(\eta)\coloneq \sigma +2ab \eta -b\eta^2$ with $\eta=x^2+y^2\ge 0$. It can be derived that $f(\eta)$ is an increasing function for $\eta\in[0, a-\gamma_\mu)$, which implies that
    \begin{align*}
        f(\eta) < f(a-\gamma_\mu)= \mu \sigma\,.
    \end{align*}
    Therefore, it follows that 
    \begin{align*}
        &\dot V(x,y) \le \mu \sigma (x^2+y^2) +\|u\|\sqrt{x^2+y^2}\\
        &\le  \varepsilon\mu \sigma (x^2+y^2) +(1-\varepsilon)\mu \sigma (x^2+y^2) +\|u\|\sqrt{x^2+y^2}\,.
    \end{align*}
    For any $\col(x,y)$ that belongs to $\mathcal{B}_{\mu}$ and satisfies 
    \begin{align*}
        \sqrt{x^2 +y^2}\ge \frac{1}{(1-\varepsilon)\mu|\sigma|}\,,
    \end{align*}
    one can derive that
    \begin{align*}
        &\dot V(x,y) \le \varepsilon\mu \sigma (x^2+y^2)\le 0\,,
    \end{align*}
    where the equality holds only when $x^2+y^2=0$. By applying the local version of Theorem~4.18 in \cite{2002_HKK_nonlinear}, one can show that the solution satisfies the inequality~\eqref{input_to_state}. 
\end{proof}

    According to Theorem~\ref{Theo_ISS}, the parameter $\sigma$ plays a pivotal role in determining the input-to-state stability of the system \eqref{forced}. A smaller $\sigma$ results in a larger value for $a-\gamma_\mu$, consequently allowing for a greater deviation of the initial condition from the origin and permitting larger perturbations (see Eqs.~\eqref{bound_initial} and \eqref{bound_inputs}). In Fig.~\ref{noisy_osci}, we compare the behavior of system \eqref{forced} at different parameter regimes. Different values of $\sigma$ result in distinct regions of attraction within the system. Even when starting from the same initial condition and being influenced by identical continuous perturbations, contrasting outcomes emerge. When the region of attraction is larger, the system remains in the neighborhood of the origin (see Fig.~\ref{noisy_osci}a); if the region of attraction is small, perturbations drive the system into undesired oscillations (see Fig.~\ref{noisy_osci}b).

\section{Networks of Bistable Oscillators}

The brain is a complex network of interacting neuronal populations. The underlying connections between the populations have a fundamental influence on the overall dynamics of the entire brain. To study the network's influence  on epileptic activity, we consider a network of bistable oscillators. To simplify the analysis, in this paper we restrict our attention to homogeneous oscillators coupled by a complete network, with dynamics governed by
\begin{subequations}\label{network}
    \begin{align}
    \dot x_k =& - \omega y_k + x_k \left(\sigma + 2 ab(x_k^2+y_k^2)-b(x_k^2+y_k^2)^2 \right) \nonumber\\
    &\hspace{4cm}+ \frac{C}{n}\sum_{\ell=1, \ell\neq k}^nx_\ell\,, \\
     \dot y_k =& \omega x_k + y_k \left(\sigma + 2 ab (x_k^2+y_k^2)-b(x_k^2+y_k^2)^2 \right) \nonumber\\
     &\hspace{4cm} +\frac{C}{n}\sum_{\ell=1, \ell\neq k}^n y_\ell\,,
\end{align}
\end{subequations}
where $k=1,2,\dots,n$, and $C\in \R$ is the coupling strength. For notational simplicity, we denote $z\coloneq[x_1,y_1,x_2,y_2,\dots,x_n,y_n]^\top$.

The parameter $\sigma$ is assumed to satisfy $-a^2/b < \sigma <0$ to ensure that each oscillator operates in the bistable regime. 
Note that our choice of additive coupling follows from recent work, which demonstrates the advantages of additive coupling over diffusive coupling \cite{2023_LMA_HK_et}.

One observes that the origin $z=0$ is still the equilibrium point of the network system \eqref{network}. Consistent with previous sections, this origin continues to represent the normal steady state in the absence of epileptic seizures. 
Next, we estimate the region of attraction to the origin and study how it is influenced by the coupling strength in the network. 

\begin{theorem}[Estimate of Region of Attraction] \label{RA_network}
For the system \eqref{network}, the following statements hold:\\
    (i) If $|C|< |\sigma|/b$, the set defined by
    \begin{align*}
        \mathcal Z \coloneq \left\{ z\in \R^{2n} \colon \|z\|^2 \le a-\sqrt{a^2+\sigma/b+|C|} \right\}
    \end{align*}
    belongs to the region of attraction for the equilibrium $z=0$. 
    That is, starting from any initial condition $z(0)\in \mathcal {Z}$, the solution $z(t)$ to \eqref{network} converges to $z=0$ asymptotically. \\
    (ii) If $|C| > |\sigma|/b$, the equilibrium $z=0$ becomes unstable. 
\end{theorem}

\begin{proof}
    \textbf{Case (i):} We construct the proof using LaSalle's invariance principle. Let $\nu \coloneq a-\sqrt{a^2+\sigma/b+|C|}$ and $z_k \coloneq [x_k,y_k]^\top$. 
    
    We first show that the set $\mathcal Z$ is invariant. To this end, consider a continuously differentiable function $V(z)=\frac{1}{2}\|z\|^2$. Then, the set $\mathcal Z$ can be equivalently rewritten as  
    \begin{equation*}
        \CZ =  \left\{ z\in \R^{2n}\colon V(x) \le 2a-2\sqrt{a^2+\sigma/b+|C|} \right\}.
    \end{equation*}    
    The time derivative of $V(x)$ satisfies
    \begin{align*}
        &\dot V(z) = \sum_{k=1}^n \Big( (x_k^2+y_k^2)\big(\sigma + 2ab (x_k^2+y_k^2)-b(x_k^2+y_k^2)^2\big) \\
        &\hspace{4cm}+\frac{C}{n} \sum_{\ell=1, \ell\neq k}^n (x_\ell x_k+y_\ell y_k )\Big)\\
         &= \sum_{k=1}^n \|z_k\|^2 (\sigma + 2ab \|z_k\|^2-b\|z_k\|^4) +\frac{C}{n}\sum_{k=1}^n \sum_{\ell=1, \ell\neq k}^n z_k^\top z_\ell.
    \end{align*}
    Observe that 
    \begin{align*}
        \sum_{k=1}^n \sum_{\ell=1, \ell\neq k}^n z_k^\top z_\ell \le \left\| \sum_{k=1}^n z_k \right\|^2\le n \sum_{k=1}^n\|z_k\|^2\,.
    \end{align*}
    Then, we arrive at
    \begin{align*}
        \dot V(z)\le \sum_{k=1}^n \|z_k\|^2 (\sigma + 2 ab\|z_k\|^2-b\|z_k\|^4) +|C|\sum_{k=1}^n\|z_k\|^2\,.
    \end{align*}
    For any $z\in \mathcal Z$, it holds that $\|z_k\|^2 \le \|z\|^2 \le \nu$. Applying the fact in the proof of Theorem~\ref{Theo_ISS} that $f(\eta)$ is increasing for $\eta\in [0,\nu]$, we have
    \begin{align*}
        \dot V(z)&\le \sum_{k=1}^n \|z_k\|^2 (\sigma + 2 ab\nu^2-b\nu^4) +|C|\sum_{k=1}^n\|z_k\|^2\\
        &=\sum_{k=1}^n \|z_k\|^2 (\sigma + 2ab \nu^2-b\nu^4+|C|)\le 0\,.
    \end{align*}
    where the last inequality is obtained by substituting $\nu = a-\sqrt{a^2 +\sigma/b +|C|}$. Then, the set $\mathcal Z$ is invariant.      
    
   Finally, observe that $\dot V(z) =0$ only when $z=0$, implying that starting in $\mathcal Z$, the solution $z(t)$ to \eqref{network} converges to $z=0$ asymptotically. 

   \textbf{Case (ii):} Linearization the system at $z=0$ leads to a Jacobian matrix, of which eigenvalues have positive real parts. Then, $z=0$ becomes unstable. 
\end{proof}

    We emphasize that Theorem~\ref{RA_network} provides an approximation of the region of attraction to the origin within the network. While the precise region of attraction may be larger than the one we have identified, finding it analytically is always challenging. Nevertheless, our estimation reveals an intriguing trend: a larger coupling strength $C$ (in amplitude) corresponds to a smaller region of attraction, indicating that an ictal state that captures seizures is more likely to emerge. It is worth noting, however, that this observation may be specific to the complete network scenario with homogeneous individual oscillators. In more complex and heterogeneous networks, the region of attraction's dependence on network structure and connection strengths can be more intricate.


    It is also worth noting that networks of bistable units can exhibit a rich repository of multi-stable dynamics. This multi-stability provides a promising framework for future exploration into partial seizures, generalized seizures, as well as the initiation and propagation of localized seizures.

\section{From Phenomenological Model to Real Data}

\begin{figure}[t]
		\centering
		\includegraphics[scale=1.05]{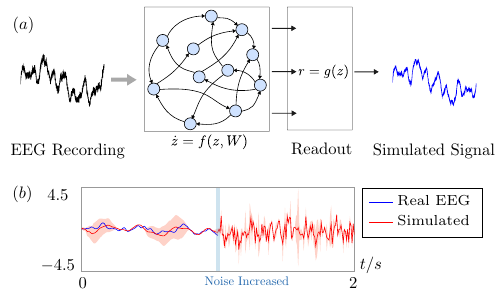}
		\caption{(a) Illustration of networks of bistable oscillators as generative models. Recorded EEG signals (left) are used to train the network parameters and the readout function so that the EEG recordings can be reproduced (right). (b) A network of 64 oscillators reproduces an EEG recording. The trained model also exhibits seizure-like activity in the presence of higher levels of noise, indicating a noise-driven transition. Shaded areas show STD.}
		\label{EEG_to_real}
\end{figure} 

When recording neural activity using EEG techniques, each electrode measures a composite signal generated by millions of neurons. Motivated by this understanding, we propose a framework to reproduce EEG data (see Fig.~\ref{EEG_to_real}a). Within this framework, we consider a network of bistable oscillators, each describing the dynamics of one neuronal ensemble. Additionally, we incorporate a nonlinear mapping to simulate the intricate journey of neural electrical activity traversing brain tissues, the skull, and ultimately reaching the EEG electrodes. Both the network parameters and the readout mapping are trainable. In Fig.~\ref{EEG_to_real}b, we train a 64-node network using seizure-free EEG recording (data source \cite{2010_SM_BM}), which accurately reproduces the real data.  Remarkably, the trained model also demonstrates seizure-like activity when subjected to higher levels of noise. 


We emphasize that this framework empowers us to delve into the hidden dynamics driving epileptic activity. Analyzing the trained network can potentially provide insights into the mechanisms of seizure onsets, which can help identify specific objectives for intervention strategies. 

\section{Concluding Remarks}
In this paper, we employ a phenomenological model to capture epileptic phenomena. At certain parameter regimes, this model exhibits bistability, where a stable equilibrium and a stable limit cycle coexist. The equilibrium signifies the normal steady-state neural activity, while the limit cycle represents pathological neural oscillations.  The transitions between these states effectively capture the onset of seizures, which can arise from internal or external perturbations. 

We have analyzed this model, demonstrating that differences in regions of attraction for both stable states distinguish epileptic brains from healthy ones. In addition, by considering a complete network, we provide analytical insights into how different connection strengths impact the region of attraction of the entire network. This analysis highlights the profound influence of network structure and connection weights on modifying the susceptibility of brain regions to noise, ultimately influencing the onset of seizures. Complementing our theoretical findings, we also discover that networks consisting of bistable units can accurately reproduce EEG recordings from epileptic brains through training parameters and network weights. This approach provides a framework to analyze the trained model and gain deeper insights into the network pathology of epilepsy.

In addition, our findings provide valuable insights into potential strategies for controlling epileptic dynamics. For instance, one approach involves designing closed-loop or event-based controllers to drive the system's states back to the region of attraction, effectively preventing them from converging into oscillatory patterns. Additionally, implementing controllers that expand the region of attraction can enhance the robustness of normal states against perturbations, thereby offering another avenue for intervention.

\appendix

\begin{pfof}{Theorem~\ref{stability}}
    We construct the proof using LaSalle's theorem.

    For Case (i), consider the Lyapunov function candidate 
    $
        V_e(x,y) = \frac{1}{2}(x^2+y^2). 
    $
    The time derivative of $V_e(x,y)$ is
    \begin{align*}
        \dot V_e(x,y) =(x^2+y^2)\left(\sigma +2a b(x^2+y^2)-b(x^2+y^2)^2 \right),
    \end{align*}
    which satisfies $\dot V_e(x,y) \le 0$ for any $\col(x,y)\in\CA_e$. Observe that the set $\CA_e$ can be equivalently written as $\CA_e=\{x,y\in\R\colon V_e(x,y) < 2{a- 2\gamma_0}\}$. Therefore, $\CA_e$ is an invariant set. In $\CA_e$, the equality $\dot V_e(x,y)=0$ holds only when $\col(x,y)=0$, which implies that starting at any initial condition in $\CA_e$, the solution to the system \eqref{main} converges to the origin asymptotically.

    For  Case (ii), the set $\CA_\ell$ can be divided into two sets $\CA_\ell'$ and $\CA_\ell''$, where
    $
        \CA_\ell'= \{x,y\in\R \colon a-\gamma_0 < x^2+y^2 \le a+\gamma_0\}, \\
        \CA_\ell''= \{ x,y\in\R \colon x^2+y^2 \ge a+\gamma_0\} 
    $
    with $\gamma_0$ defined as before. Consider the Lyapunov function candidates 
    \begin{align*}
        &V_\ell'(x,y) = a+\gamma_0 - (x^2+y^2),  \\&V_\ell''(x,y) =  (x^2+y^2) - a-\gamma_0. 
    \end{align*}
    It can be easily shown that $\dot V_\ell'(x,y) \le 0$ for any $\col(x,y)\in\CA_\ell'$ and $\dot V_\ell''(x,y) \le 0$ for any $\col(x,y)\in\CA_\ell''$. As $\CA_\ell'$ and $\CA_\ell''$ can be written as $\CA_\ell'=\{ x,y\in \R\colon 0 \le V_\ell'(x,y)< 2\gamma_0\}$ and $\CA_\ell'=\{ x,y\in \R\colon V_\ell'(x,y)\ge 0 \}$, one can observe that both $\CA_\ell'$ and $\CA_\ell''$ are invariant. Therefore, the set $\CA_\ell=\CA_\ell'\cup \CA_\ell''$ is invariant, and $\dot V(x,y)=0$ only when $\col(x,y)\in \CC_1$. Starting from any initial condition in $\CA_\ell$, the solution converges to $\CC_1$ asymptotically.
    \end{pfof}

\bibliographystyle{IEEEtran}
\bibliography{refers}

\end{document}